\documentclass[12pt,a4paper]{article}
\usepackage{geometry}
\usepackage{amsmath,amsthm}
\usepackage{threeparttable}
\usepackage{array}
\usepackage{appendix}
\usepackage{natbib}
\usepackage{graphicx}
\usepackage[english]{babel}
\theoremstyle{thmsty}
\usepackage{amssymb}
\newtheorem{thm}{Theorem}%[section]
\theoremstyle{definition}
\allowdisplaybreaks

\begin{document}
\fontsize{19}{18pt}\selectfont
\begin{center}
\textbf{An aberration criterion for conditional models}
\end{center}
\fontsize{10}{16pt}\selectfont

\fontsize{12}{20pt}\selectfont

\begin{center}
Ming-Chung Chang\\
National Central University
\end{center}

\begin{abstract}
Conditional models with one pair of conditional and conditioned factors in \citet{Mukerjee:17} are extended to two pairs in this paper.
The extension includes the parametrization, effect hierarchy, sufficient conditions for universal optimality, aberration, complementary set theory and the strategy for finding minimum aberration designs.
A catalogue of 16-run minimum aberration designs under conditional models is provided.
For five to twelve factors, all 16-run minimum aberration designs under conditional models are also minimum aberration under traditional models.


\end{abstract}

\section{Introduction}
Two-level factorial designs are widely applicable in diverse fields especially when the purpose of experiments is factor screening.
There has been rich literature on the evaluation of such designs.
Minimum aberration, proposed in \citet{Fries:80}, is the most popular criterion under model uncertainty.
We refer readers to \citet{Mukerjee:06}, \citet{Jeff:09} and \citet{Cheng:14} for further references.

There is a situation that the factorial effects of one factor are more meaningful to be defined conditionally on each fixed level of another factor.
Sliding level experiments in engineering are of the case \citep[p.~346]{Jeff:09}.
\citet{Jeff:14} provided more examples.
Another situation is that a tradition model is used first and then transformed to a conditional model for specific purposes.
Refer to the de-aliasing method in \citet{Su:15} for more detail.
The model under this kind of experiments is referred to as a \emph{conditional model} in \citet{Mukerjee:17}.
Under this model, \citet{Mukerjee:17} proposed a new minimum aberration criterion and provided catalogues for 16-run and 32-run minimum aberration designs when there is only one pair of conditional and conditioned factors.

In this paper, we consider two-level factorials with two pairs of conditional and conditioned factors.
For each pair, the main effect and interaction effects involving one factor are defined conditionally on each fixed level of the other factor.
The former is called a \emph{conditional factor} and the latter is called a \emph{conditioned factor}.
The remaining factors are called \emph{traditional factors}.
We extend the parametrization, effect hierarchy order, aberration and complementary set theory in \citet{Mukerjee:17} to this setting.
Catalogues of $16$-run and $32$-run minimum aberration designs under a condition model are also provided.

The paper is organized as follows.
Section \ref{Sec: Parametrization and Effect hierarchy} introduces the parametrization and a new hierarchy order of effects under a conditional model.
Some sufficient conditions for a design to be universally optimal under a conditional model with only main effects are given in Section \ref{Sec: Universally optimal designs in the absence of interactions}.
Section \ref{Sec: Minimum aberration criterion} introduces a new minimum aberration criterion.
Section \ref{Sec: Regular designs: complementary set thoery} develops a complementary set theory, which is useful for finding minimum aberration designs for large number of factors.
An efficient computational procedure is developed in Section \ref{Sec: An efficient computation procedure} to search for minimum aberration designs for any number of factors.
Section \ref{Sec: Concluding remarks} concludes our work.

\section{Parametrization and Effect hierarchy}\label{Sec: Parametrization and Effect hierarchy}
\subsection{Parametrization}
Consider a $2^n$ factorial design with $n$ ($\geq 4$) factors $F_{1},...,F_{n}$, each at levels $0$ and $1$.
The $n$ factors consist of two pairs of conditional and conditioned factors, say $F_{1}, F_{2}$ and $F_{3}, F_{4}$.
The main effect and interaction effects involving $F_{1}$ (respectively, $F_{3}$) are defined conditionally on each fixed level of $F_{2}$ (respectively, $F_{4}$).
Define $\Omega$ as the set of $\nu=2^n$ binary $n$-tuples.
For $(i_1,...,i_n)\in\Omega$, let $\tau(i_1\cdots i_n)$ be the treatment effect of treatment combination $i_1\cdots i_n$.
Similarly, we write $\theta(i_1\cdots i_n)$ for the factorial effect $F_{1}^{i_1}\cdots F_{n}^{i_n}$ when $i_1\cdots i_n$ is nonnull, and $\theta(0\cdots 0)$ for the grand mean.
Let $\boldsymbol{\tau}$ and $\boldsymbol{\theta}$ be $\nu\times 1$ vectors with elements $\tau(i_1\cdots i_n)$ and $\theta(i_1\cdots i_n)$ arranged in the lexicographic order, respectively.
Then the traditional full factorial model is given by
\begin{align}\label{Relation_tau_theta}
\boldsymbol{\tau}=\mathbf{H}^{\otimes n}\boldsymbol{\theta},
\end{align}
where $\otimes$ represents the Kronecker product and $\mathbf{H}^{\otimes n}$ denotes the $n$-fold Kronecker product of
\begin{align*}
\mathbf{H}=\left( \begin{array}{cc}
1 & 1 \\
1 & -1  \end{array} \right),
\end{align*}
a Hadamard matrix of order two.
Let $\mathbf{H}(0)=(1,1)$ and $\mathbf{H}(1)=(1,-1)$ be the top and bottom rows of $\mathbf{H}$, respectively.

Let $\beta(j_1\cdots j_n)$ be the factorial effect $F_{1}^{i_1}\cdots F_{n}^{i_n}$ under a conditional model with conditional and conditioned factors $F_{1},F_{2}$ and $F_{3},F_{4}$, respectively.
Denote the vector with the $\nu$ elements $\beta(j_1\cdots j_n)$'s by $\boldsymbol{\beta}$.
Because there are two pairs of conditional and conditioned factors, we can reparametrize $\boldsymbol{\theta}$ by $\boldsymbol{\beta}$ with
\begin{align}\label{Relation_beta_tau}
\boldsymbol{\beta}=\nu^{-1}\mathbf{W}\otimes \mathbf{H}^{\otimes (n-4)}\boldsymbol{\tau},
\end{align}
where
\begin{align*}
\mathbf{W}=\left( \begin{array}{l}
\mathbf{H}(0) \otimes \mathbf{H} \otimes \mathbf{H}(0) \otimes \mathbf{H}\\
\mathbf{H}(1) \otimes \sqrt{2}\mathbf{I}_{2}  \otimes\mathbf{H}(0) \otimes \mathbf{H}\\
\mathbf{H}(0) \otimes \mathbf{H}  \otimes\mathbf{H}(1) \otimes \sqrt{2}\mathbf{I}_{2}\\
\mathbf{H}(1) \otimes \sqrt{2}\mathbf{I}_{2}  \otimes \mathbf{H}(1) \otimes \sqrt{2}\mathbf{I}_{2}
\end{array} \right),
\end{align*}
in which $\mathbf{I}_{2}$ is the identity matrix of order two.
We cluster $\boldsymbol{\beta}$ into vectors representing unconditional and conditional factorial effects of various orders.
Define
\begin{align*}
\Omega_{0l}=&\{ (j_{1},...,j_{n}): j_{1}=j_{3}=0, \text{ and $l$ of $j_{2}, j_{4},...,j_{n}$ equal 1} \}, \\
\Omega_{1l}=&\{ (j_{1},...,j_{n}): j_{1}=1, j_{2}=0,1, j_{3}=0, \text{ and $l-1$ of $j_{4},...,j_{n}$ equal 1} \} \\
&\cup\{ (j_{1},...,j_{n}): j_{3}=1, j_{4}=0,1, j_{1}=0, \text{ and $l-1$ of $j_{2},j_{5},...,j_{n}$ equal 1} \}, \\
\Omega_{2l}=&\{ (j_{1},...,j_{n}): j_{1}=j_{3}=1, j_{2}=0,1, j_{4}=0,1, \text{ and $l-2$ of $j_{5},...,j_{n}$ equal 1} \},
\end{align*}
where $1\leq l\leq n-2$.
Let $\boldsymbol{\beta}_{sl}$ be the vector with elements $\beta(j_1\cdots j_n)$, where $(j_1,...,j_n)\in\Omega_{sl}$.

\subsection{Effect hierarchy}
The effect hierarchy of the $\beta(j_1\cdots j_n)$'s is defined via a prior specification on $\boldsymbol{\tau}$ in terms of a zero-mean Gaussian random function such that $\text{cov}(\boldsymbol{\tau})=\sigma^2\mathbf{R}^{\otimes n}$, where $\sigma^2>0$ and the $2\times 2$ matrix $\mathbf{R}$ has diagonal elements $1$ and off-diagonal elements $\rho$, $0<\rho<1$.
Then the prior covariance matrix of $\boldsymbol{\beta}$ is given by
\begin{align*}
\text{cov}(\boldsymbol{\beta})
&=\nu^{-2}\{\mathbf{W}\otimes\mathbf{H}^{\otimes (n-4)}\}\text{cov}(\boldsymbol{\tau})\{\mathbf{W}\otimes\mathbf{H}^{\otimes (n-4)}\}^T\\
&=\sigma^2\nu^{-2}\{\mathbf{W}\mathbf{R}^{\otimes 4} \mathbf{W}^T\}\otimes\{\mathbf{H}\mathbf{R}\mathbf{H}\}^{\otimes (n-4)}.
\end{align*}
The following theorem gives the variances of the $\beta(j_1\cdots j_n)$'s.
\begin{thm}\label{Thm: prior_var}
For a $(j_{1},...,j_{n})\in \Omega_{sl}$, we have
\[ \text{var}(\beta(j_{1}\cdots j_{n}))=\sigma^2\nu^{-1}(1+\rho)^{n-l-s}(1-\rho)^{l}. \]
\end{thm}
\begin{proof}
Recall that $\text{cov}(\boldsymbol{\beta})
=\sigma^2\nu^{-2}\{\mathbf{W}\mathbf{R}^{\otimes 4} \mathbf{W}^T\}\otimes\{\mathbf{H}\mathbf{R}\mathbf{H}\}^{\otimes (n-4)}$.
It can be easily verified that $\mathbf{H}\mathbf{R}\mathbf{H}=2\text{diag}(1+\rho,1-\rho)$,
\begin{align*}
\mathbf{W}=\left( \begin{array}{cccc}
\mathbf{H}^{\otimes 2} & \mathbf{H}^{\otimes 2} & \mathbf{H}^{\otimes 2} & \mathbf{H}^{\otimes 2} \\
\sqrt{2}\mathbf{I}_{2}\otimes \mathbf{H} & \sqrt{2}\mathbf{I}_{2}\otimes \mathbf{H} & -\sqrt{2}\mathbf{I}_{2}\otimes \mathbf{H} & -\sqrt{2}\mathbf{I}_{2}\otimes \mathbf{H} \\
\sqrt{2}\mathbf{H}\otimes \mathbf{I}_{2}  & -\sqrt{2}\mathbf{H}\otimes \mathbf{I}_{2} & \sqrt{2}\mathbf{H}\otimes \mathbf{I}_{2} & -\sqrt{2}\mathbf{H}\otimes \mathbf{I}_{2} \\
2\mathbf{I}_{2}^{\otimes 2} & -2\mathbf{I}_{2}^{\otimes 2} & -2\mathbf{I}_{2}^{\otimes 2} & 2\mathbf{I}_{2}^{\otimes 2}
\end{array} \right),
\end{align*}
and
\begin{align*}
\mathbf{R}^{\otimes 4}=\left( \begin{array}{cccc}
\mathbf{R}^{\otimes 2} & \rho\mathbf{R}^{\otimes 2} & \rho\mathbf{R}^{\otimes 2} & \rho^2\mathbf{R}^{\otimes 2} \\
\rho\mathbf{R}^{\otimes 2} & \mathbf{R}^{\otimes 2} & \rho^2\mathbf{R}^{\otimes 2} & \rho\mathbf{R}^{\otimes 2} \\
\rho\mathbf{R}^{\otimes 2} & \rho^2\mathbf{R}^{\otimes 2} & \mathbf{R}^{\otimes 2} & \rho\mathbf{R}^{\otimes 2} \\
\rho^2\mathbf{R}^{\otimes 2} & \rho\mathbf{R}^{\otimes 2} & \rho\mathbf{R}^{\otimes 2} & \mathbf{R}^{\otimes 2}
\end{array} \right).
\end{align*}
The $\mathbf{W}\mathbf{R}^{\otimes 4} \mathbf{W}^T$ is a $16\times 16$ matrix, which can be regarded as a $4\times 4$ block-matrix with each block being a $4\times 4$ matrix.
Denote the $(i,j)$th block-matrix of $\mathbf{W}\mathbf{R}^{\otimes 4} \mathbf{W}^T$ by $[\mathbf{W}\mathbf{R}^{\otimes 4} \mathbf{W}^T]_{ij}$.
Then by calculation, we get $[\mathbf{W}\mathbf{R}^{\otimes 4} \mathbf{W}^T]_{ij}=\mathbf{0}$ when $i\neq j$; for $i=j$, we obtain $[\mathbf{W}\mathbf{R}^{\otimes 4} \mathbf{W}^T]_{11}=4(1+\rho)^2(\mathbf{H}\mathbf{R}\mathbf{H})^{\otimes 2}$,
$[\mathbf{W}\mathbf{R}^{\otimes 4} \mathbf{W}^T]_{22}=8(1-\rho^2)\mathbf{R}\otimes (\mathbf{H}\mathbf{R}\mathbf{H})$,
$[\mathbf{W}\mathbf{R}^{\otimes 4} \mathbf{W}^T]_{33}=8(1-\rho^2)(\mathbf{H}\mathbf{R}\mathbf{H})\otimes \mathbf{R}$ and
$[\mathbf{W}\mathbf{R}^{\otimes 4} \mathbf{W}^T]_{44}=16(1-\rho)^2 \mathbf{R} \otimes \mathbf{R}$.

The diagonal elements of $(\mathbf{H}\mathbf{R}\mathbf{H})^{\otimes 2}$, denoted by $\text{diag}((\mathbf{H}\mathbf{R}\mathbf{H})^{\otimes 2})$, can be obtained as $4(1+\rho)^2,4(1-\rho)(1+\rho),4(1-\rho)(1+\rho),4(1-\rho)^2$.
We also have $\text{diag}(\mathbf{R}\otimes (\mathbf{H}\mathbf{R}\mathbf{H}))=2(1+\rho,1-\rho,1+\rho,1-\rho)$,
$\text{diag}((\mathbf{H}\mathbf{R}\mathbf{H})\otimes \mathbf{R})=2(1+\rho,1+\rho,1-\rho,1-\rho)$ and
$\text{diag}(\mathbf{R}\otimes \mathbf{R})=(1,1,1,1)$.
Thus, we get $\text{diag}([\mathbf{W}\mathbf{R}^{\otimes 4} \mathbf{W}^T]_{11})=16((1+\rho)^4,(1-\rho)(1+\rho)^3,(1-\rho)(1+\rho)^3,(1-\rho)^2(1+\rho)^2)$,
$\text{diag}([\mathbf{W}\mathbf{R}^{\otimes 4} \mathbf{W}^T]_{22})=16((1+\rho)^2(1-\rho),(1+\rho)(1-\rho)^2,(1+\rho)^2(1-\rho),(1+\rho)(1-\rho)^2)$,
$\text{diag}([\mathbf{W}\mathbf{R}^{\otimes 4} \mathbf{W}^T]_{33})=16((1+\rho)^2(1-\rho),(1+\rho)^2(1-\rho),(1+\rho)(1-\rho)^2,(1+\rho)(1-\rho)^2)$ and
$\text{diag}([\mathbf{W}\mathbf{R}^{\otimes 4} \mathbf{W}^T]_{44})=16((1-\rho)^2,(1-\rho)^2,(1-\rho)^2,(1-\rho)^2)$.

Because the variances are only related to the diagonal elements of $\text{cov}(\boldsymbol{\beta})$, one can easily check for a $(j_{1},...,j_{n})\in \Omega_{sl}$,
$ \text{var}(\beta(j_{1}\cdots j_{n}))=\sigma^2\nu^{-1}(1+\rho)^{n-l-s}(1-\rho)^{l} $
based on the above calculation.
\end{proof}

Denote the variances of $\beta(j_{1}\cdots j_{n})$, where $(j_{1},...,j_{n})\in \Omega_{sl}$, by $V_{sl}$.
From Theorem \ref{Thm: prior_var}, it is clear that $V_{0l}>V_{1l}>V_{2l}$ for $2\leq l\leq n-2$.
Because $V_{2l}/V_{0,l+1}=1/(1-\rho^2)>1$ for all $0<\rho<1$, we have
\begin{align}\label{prior_hierarchy}
V_{01}>V_{11}>V_{02}>V_{12}>V_{22}>V_{03}>V_{13}>V_{23}>\cdots >V_{0,n-2}>V_{1,n-2}>V_{2,n-2}.
\end{align}
In view of (\ref{prior_hierarchy}), we define the following effect hierarchy under a conditional model as follows.
Under a conditional model, the unconditional main effects are the most important, while the conditional main effects are positioned next; then come the unconditional two-factor interactions, followed by the one-pair conditional two-factor interactions, then two-pair conditional two-factor interactions, and so on.

\section{Universally optimal designs in the absence of interactions}\label{Sec: Universally optimal designs in the absence of interactions}
\subsection{Linking the traditional and conditional models}
The connection between conditional effects $\boldsymbol{\beta}$ and traditional factorial effects $\boldsymbol{\theta}$ can be established by (\ref{Relation_tau_theta}) and (\ref{Relation_beta_tau}) as follows:
\begin{align*}
\boldsymbol{\beta}
&=\nu^{-1}\mathbf{W}\otimes \mathbf{H}^{\otimes (n-4)}\boldsymbol{\tau}\\
&=\nu^{-1}\mathbf{W}\otimes \mathbf{H}^{\otimes (n-4)}\mathbf{H}^{\otimes n}\boldsymbol{\theta}\\
&=\nu^{-1}\{\mathbf{W}\mathbf{H}^{\otimes 4}\}\otimes \{\mathbf{H}\mathbf{H}\}^{\otimes (n-4)}\boldsymbol{\theta}.
\end{align*}
By using the fact $\mathbf{H}\mathbf{H}=2\mathbf{I}_{2}$ and
\begin{align*}
\mathbf{W}\mathbf{H}^{\otimes 4}=
4\left( \begin{array}{cccc}
\mathbf{H}^{\otimes 2}\mathbf{H}^{\otimes 2} & \boldsymbol{0} & \boldsymbol{0} & \boldsymbol{0} \\
\boldsymbol{0} & \boldsymbol{0} & \sqrt{2}\mathbf{I}_{2}\otimes\mathbf{H}\mathbf{H}^{\otimes 2} & \boldsymbol{0} \\
\boldsymbol{0} & \sqrt{2}\mathbf{H}\otimes\mathbf{I}_{2}\mathbf{H}^{\otimes 2} & \boldsymbol{0} & \boldsymbol{0} \\
\boldsymbol{0} & \boldsymbol{0} & \boldsymbol{0} & 2\mathbf{I}_{2}^{\otimes 2}\mathbf{H}^{\otimes 2}
\end{array} \right),
\end{align*}
we obtain
\begin{align*}
\boldsymbol{\beta}=
\left( \begin{array}{cccc}
\mathbf{I}_{2}^{\otimes (n-2)} & \boldsymbol{0} & \boldsymbol{0} & \boldsymbol{0} \\
\boldsymbol{0} & \boldsymbol{0} & \frac{1}{\sqrt{2}}\mathbf{H}\otimes\mathbf{I}_{2}\otimes \mathbf{I}_{2}^{\otimes (n-4)} & \boldsymbol{0} \\
\boldsymbol{0} & \frac{1}{\sqrt{2}}\mathbf{I}_{2}\otimes\mathbf{H}\otimes \mathbf{I}_{2}^{\otimes (n-4)} & \boldsymbol{0} & \boldsymbol{0} \\
\boldsymbol{0} & \boldsymbol{0} & \boldsymbol{0} & \frac{1}{2}\mathbf{H}^{\otimes 2}\otimes \mathbf{I}_{2}^{\otimes (n-2)}
\end{array} \right)
\boldsymbol{\theta},
\end{align*}
which implies
\begin{align*}
\boldsymbol{\theta}=
\left( \begin{array}{cccc}
\mathbf{I}_{2}^{\otimes (n-2)} & \boldsymbol{0} & \boldsymbol{0} & \boldsymbol{0} \\
\boldsymbol{0} & \boldsymbol{0} & \frac{1}{\sqrt{2}}\mathbf{H}\otimes\mathbf{I}_{2}\otimes \mathbf{I}_{2}^{\otimes (n-4)} & \boldsymbol{0} \\
\boldsymbol{0} & \frac{1}{\sqrt{2}}\mathbf{I}_{2}\otimes\mathbf{H}\otimes \mathbf{I}_{2}^{\otimes (n-4)} & \boldsymbol{0} & \boldsymbol{0} \\
\boldsymbol{0} & \boldsymbol{0} & \boldsymbol{0} & \frac{1}{2}\mathbf{H}^{\otimes 2}\otimes \mathbf{I}_{2}^{\otimes (n-2)}
\end{array} \right)
\boldsymbol{\beta}
\end{align*}
because $\mathbf{H}^{-1}=(1/2)\mathbf{H}$.
This yields
\begin{align*}
&\theta(0j_{2}0j_{4}j_{5}\cdots j_{n})=\beta(0j_{2}0j_{4}j_{5}\cdots j_{n}),\\
&\theta(1j_{2}0j_{4}j_{5}\cdots j_{n})=\frac{1}{\sqrt{2}}\left\{\beta(100j_{4}j_{5}\cdots j_{n})+\delta(j_{2})\beta(110j_{4}j_{5}\cdots j_{n})\right\},\\
&\theta(0j_{2}1j_{4}j_{5}\cdots j_{n})=\frac{1}{\sqrt{2}}\left\{\beta(0j_{2}10j_{5}\cdots j_{n})+\delta(j_{4})\beta(0j_{2}11j_{5}\cdots j_{n})\right\},\\
&\theta(1j_{2}1j_{4}j_{5}\cdots j_{n})=\frac{1}{2}\{\beta(1010j_{5}\cdots j_{n})+\delta(j_{4})\beta(1011j_{5}\cdots j_{n})\\
&\hspace{3.5cm}+\delta(j_{2})\beta(1110j_{5}\cdots j_{n})+\delta(j_{2})\delta(j_{4})\beta(1111j_{5}\cdots j_{n})\},
\end{align*}
where $\delta(j)=-2j+1$.

Consider an $N$-run design represented by $\mathbf{D}$, an $N\times n$ matrix with $1$ being the high level and $-1$ being the low level.
Denote the corresponding $N\times 2^n$ full model matrix under a traditional model by $\mathbf{X}$.
Denote the vector of responses by $\boldsymbol{y}$.
Then we have $E(\boldsymbol{y})=\mathbf{X}\boldsymbol{\theta}$.
Each column of $\mathbf{X}$ is represented by $\boldsymbol{x}(j_{1}\cdots j_{n})$, $(j_{1},...,j_{n})\in\Omega$.
Let
\begin{align*}
&\boldsymbol{z}(0j_{2}0j_{4}j_{5}\cdots j_{n})=\boldsymbol{x}(0j_{2}0j_{4}j_{5}\cdots j_{n}),\\
&\boldsymbol{z}(1j_{2}0j_{4}j_{5}\cdots j_{n})=\frac{1}{\sqrt{2}}\left\{\boldsymbol{x}(100j_{4}j_{5}\cdots j_{n})+\delta(j_{2})\boldsymbol{x}(110j_{4}j_{5}\cdots j_{n})\right\},\\
&\boldsymbol{z}(0j_{2}1j_{4}j_{5}\cdots j_{n})=\frac{1}{\sqrt{2}}\left\{\boldsymbol{x}(0j_{2}10j_{5}\cdots j_{n})+\delta(j_{4})\boldsymbol{x}(0j_{2}11j_{5}\cdots j_{n})\right\},\\
&\boldsymbol{z}(1j_{2}1j_{4}j_{5}\cdots j_{n})=\frac{1}{2}\{\boldsymbol{x}(1010j_{5}\cdots j_{n})+\delta(j_{4})\boldsymbol{x}(1011j_{5}\cdots j_{n})\\
&\hspace{3.5cm}+\delta(j_{2})\boldsymbol{x}(1110j_{5}\cdots j_{n})+\delta(j_{2})\delta(j_{4})\boldsymbol{x}(1111j_{5}\cdots j_{n})\},
\end{align*}
where $\delta(j)=-2j+1$.
Let $\mathbf{Z}_{sl}$ and $\mathbf{X}_{sl}$ consist of the $\boldsymbol{z}(j_{1}\cdots j_{n})$'s and $\boldsymbol{x}(j_{1}\cdots j_{n})$'s respectively, where $(j_{1},...,j_{n})\in\Omega_{sl}$.
Then the conditional model under $\mathbf{D}$ is
\begin{align}\label{Cond_model}
E(\boldsymbol{y})=\boldsymbol{z}(0\cdots 0)\beta(0\cdots 0)+\sum_{s=0}^{2}\sum_{l=1}^{n-2}\mathbf{Z}_{sl}\boldsymbol{\beta}_{sl}.
\end{align}

\subsection{Universally optimal designs}
If all interactions are assumed to be absent, then the model (\ref{Cond_model}) reduces to
\begin{align}\label{Cond_model_ME}
E(\boldsymbol{y})=\boldsymbol{z}(0\cdots 0)\beta(0\cdots 0)+\mathbf{Z}_{01}\boldsymbol{\beta}_{01}+\mathbf{Z}_{11}\boldsymbol{\beta}_{11}.
\end{align}
In the following, we present a theorem which gives some requirements for a design to be universally optimal under model (\ref{Cond_model_ME}).
\begin{thm}\label{Thm: universally_optimal}
Suppose an $N$-run design $\mathbf{D}$ satisfies
\begin{enumerate}
\item[(i).] $\mathbf{D}$ is an orthogonal array of strength two;
\item[(ii).] all eight triples of symbols occur equally often when $\mathbf{D}$ is projected onto $F_{1},F_{2},F_{j}$, $j\in\{4,5,...,n\}$;
\item[(iii).] all eight triples of symbols occur equally often when $\mathbf{D}$ is projected onto $F_{3},F_{4},F_{j}$, $j\in\{2,5,...,n\}$;
\item[(iv).] all sixteen triples of symbols occur equally often when $\mathbf{D}$ is projected onto $F_{1},F_{2},F_{3},F_{4}$.
\end{enumerate}
Then $\mathbf{D}$ is universally optimal among all $N$-run designs for inference on $\boldsymbol{\beta}_{01}$ and $\boldsymbol{\beta}_{11}$ under model (\ref{Cond_model_ME}).
\end{thm}
\begin{proof}
Let $\mathbf{Z}_{1}=(\mathbf{Z}_{01},\mathbf{Z}_{11})$.
Denote the information matrix of $\boldsymbol{\beta}_{01}$ and
$\boldsymbol{\beta}_{11}$ under model (\ref{Cond_model_ME}) by $\mathbf{M}$.
Because $\mathbf{Z}_{1}^T\mathbf{Z}_{1}-\mathbf{M}$ is nonnegative definite, we have
\begin{align}\label{Ineq: infor}
\text{tr}[\mathbf{M}]\leq \text{tr}[\mathbf{Z}_{1}^T\mathbf{Z}_{1}]=N(n+1)
\end{align}
for every $N$-run design.
Note that $\mathbf{M}$ can be obtained by
\begin{align*}
\mathbf{M}=\mathbf{Z}_{1}^T\{\mathbf{I}_{N}-\boldsymbol{z}(0\cdots 0)[\boldsymbol{z}(0\cdots 0)^T\boldsymbol{z}(0\cdots 0)]^{-1}\boldsymbol{z}(0\cdots 0)^T\}\mathbf{Z}_{1},
\end{align*}
which can be simplified as $\mathbf{M}=\mathbf{Z}_{1}^T\{\mathbf{I}_{N}-\frac{1}{N}\boldsymbol{1}_{N}\boldsymbol{1}_{N}^T\}\mathbf{Z}_{1}$ because $\boldsymbol{z}(0\cdots 0)=\boldsymbol{1}_{N}$.
Under the conditions (i),...,(iv), it is easy to verify that $\mathbf{Z}_{1}^T\boldsymbol{1}_{N}=\boldsymbol{0}$ and $\mathbf{Z}_{1}^T\mathbf{Z}_{1}=N\mathbf{I}_{n+1}$.
Thus $\mathbf{M}=N\mathbf{I}_{n+1}$ and $\text{tr}[\mathbf{M}]$ reaches the upper bound in (\ref{Ineq: infor}).
The result now follows from \citet{Kiefer:75}.
\end{proof}

\section{Minimum aberration criterion}\label{Sec: Minimum aberration criterion}
Set $n\geq 4$ to avoid trivialities.
We consider designs meeting (i),...,(iv) of Theorem \ref{Thm: universally_optimal}.
Under model (\ref{Cond_model_ME}), $\widehat{\boldsymbol{\beta}}_{h1}=N^{-1}\mathbf{Z}_{h1}^T\boldsymbol{y}$ is the best linear unbiased estimator of $\boldsymbol{\beta}_{h1}$, $h=0,1$.
We revert to the full model (\ref{Cond_model}) to assess the impact of possible presence of interactions on $\widehat{\boldsymbol{\beta}}_{h1}$.
Under model (\ref{Cond_model}), $\widehat{\boldsymbol{\beta}}_{h1}$ is no longer unbiased but has bias $N^{-1}\sum_{s=0}^{2}\sum_{l=2}^{n-2}\mathbf{Z}_{h1}^T\mathbf{Z}_{sl}\boldsymbol{\beta}_{sl}$.
A reasonable measure of the bias in $\widehat{\boldsymbol{\beta}}_{h1}$ caused by the interactions, as in \citet{Tang:99}, is
\begin{align*}
K_{sl}(h)=N^{-2}\text{tr}[\mathbf{Z}_{h1}^T\mathbf{Z}_{sl}\mathbf{Z}_{sl}^T\mathbf{Z}_{h1}]=N^{-2}\text{tr}[\mathbf{X}_{h1}^T\mathbf{X}_{sl}\mathbf{X}_{sl}^T\mathbf{X}_{h1}],
\end{align*}
where the last equality holds because $\mathbf{X}_{sl}$ is an orthogonal transform of $\mathbf{Z}_{sl}$.

Based on the effect hierarchy in (\ref{prior_hierarchy}), a minimum aberration design is one which minimizes the terms of
\begin{align}\label{Seq_K}
K=\{ K_{02}(0),K_{02}(1),K_{12}(0),K_{12}(1),K_{22}(0),K_{22}(1),K_{03}(0),K_{03}(1),... \}
\end{align}
in a sequential manner from left to right.

\section{Regular designs: complementary set thoery}\label{Sec: Regular designs: complementary set thoery}
We now focus on regular designs under the conditional model.
Let $\Delta_{r}$ be the set of nonnull $r\times 1$ binary vectors.
All operations with these vectors are over the finite field GF(2).
A regular design in $N=2^r$ ($r<n$) runs is given by $n$ distinct vectors $b_{1},...,b_{n}$ from $\Delta_{r}$ such that the matrix $B=(b_{1},...,b_{n})$ has full row rank.
The design consists of the $N$ treatment combinations in the row space of $B$.

In the following, we define some useful quantities to represent $K_{sl}(h)$.
Let $A_{l}^{(1)}$ be the number of ways of choosing $l$ out of $b_{2},b_{4},...,b_{n}$ such that the sum of the chosen $l$ equals $0$.
Let $A_{l}^{(21)}$ be the number of ways of choosing $l$ out of $b_{4},...,b_{n}$ such that the sum of the chosen $l$ is in the set $\{b_{1},b_{1}+b_{2}\}$.
Let $A_{l}^{(22)}$ be the number of ways of choosing $l$ out of $b_{2},b_{5},...,b_{n}$ such that the sum of the chosen $l$ is in the set $\{b_{3},b_{3}+b_{4}\}$.
Let $A_{l}^{(2)}=A_{l}^{(21)}+A_{l}^{(22)}$.
Let $A_{l}^{(31)}$ be the number of ways of choosing $l$ out of $b_{4},...,b_{n}$ such that the sum of the chosen $l$ is in the set $\{0,b_{2}\}$.
Let $A_{l}^{(32)}$ be the number of ways of choosing $l$ out of $b_{2},b_{5},...,b_{n}$ such that the sum of the chosen $l$ is in the set $\{0,b_{4}\}$.
Let $A_{l}^{(3)}=A_{l}^{(31)}+A_{l}^{(32)}$.
Let $A_{l}^{(42)}$ be the number of ways of choosing $l$ out of $b_{2},b_{5},...,b_{n}$ such that the sum of the chosen $l$ is in the set $\{b_{1}+b_{3},b_{1}+b_{3}+b_{4}\}$.
Let $A_{l}^{(43)}$ be the number of ways of choosing $l$ out of $b_{5},...,b_{n}$ such that the sum of the chosen $l$ is in the set $\{b_{1}+b_{3},b_{1}+b_{3}+b_{4}\}$.
Let $A_{l}^{(52)}$ be the number of ways of choosing $l$ out of $b_{5},...,b_{n}$ such that the sum of the chosen $l$ is in the set $\{b_{1}+b_{2}+b_{3},b_{1}+b_{2}+b_{3}+b_{4}\}$.
Let $A_{l}^{(7)}$ be the number of ways of choosing $l$ out of $b_{5},...,b_{n}$ such that the sum of the chosen $l$ is in the set $\{b_{1}+b_{3},b_{1}+b_{2}+b_{3},b_{1}+b_{3}+b_{4},b_{1}+b_{2}+b_{3}+b_{4}\}$.
Let $A_{l}^{(8)}$ be the number of ways of choosing $l$ out of $b_{5},...,b_{n}$ such that the sum of the chosen $l$ is in the set $\{b_{1},b_{3},b_{1}+b_{2},b_{1}+b_{4},b_{2}+b_{3},b_{3}+b_{4},b_{1}+b_{2}+b_{4},b_{2}+b_{3}+b_{4}\}$.
The next result gives expressions for $K_{sl}(h)$ in terms of the quantities just introduced.
\begin{thm}
For $2\leq l\leq n-2$, we have
\begin{enumerate}
\item [(a).] $K_{0l}(0)=(l+1)A_{l+1}^{(1)}+(n-l-1)A_{l-1}^{(1)}$;
\item [(b).] $K_{0l}(1)=A_{l-1}^{(2)}+A_{l}^{(2)}$;
\item [(c).] $K_{1l}(0)=(n-l-1)A_{l-2}^{(2)} +A_{l-1}^{(2)}+lA_{l}^{(2)}$;
\item [(d).] $K_{1l}(1)=2A_{l-1}^{(3)}+2\{A_{l-1}^{(42)}+A_{l-2}^{(43)}+A_{l-1}^{(52)}\}$;
\item [(e).] $K_{2l}(0)=2A_{l-2}^{(7)}+(n-l-1)A_{l-3}^{(7)}+(l-1)A_{l-1}^{(7)}$;
\item [(f).] $K_{2l}(1)=2A_{l-2}^{(8)}$.
\end{enumerate}
\end{thm}
\begin{proof}
In the proof, a traditional factorial effect is represented by a word, i.e., a subset of $\{1,...,n\}$.
For two words $W_{1}$ and $W_{2}$, we define $W_{1}\triangle W_{2}$ to be $(W_{1}\cup W_{2})\setminus (W_{1}\cap W_{2})$.
Note that $K_{sl}(h)=N^{-2}\text{tr}[\mathbf{X}_{h1}^T\mathbf{X}_{sl}\mathbf{X}_{sl}^T\mathbf{X}_{h1}]$, which is the sum of squared entries of $N^{-2}\mathbf{X}_{h1}^T\mathbf{X}_{sl}$.
Because the design is regular, each squared entry is either one or zero according to whether the corresponding effects are aliased.

Part (a) is evident from \citet{Tang:99} except that the number of factors considered in the computation is $n-2$ (exclude $F_{1}$ and $F_{3}$).
So we have $K_{0l}(0)=(l+1)A_{l+1}^{(1)}+(n-l-1)A_{l-1}^{(1)}$.

For (b), let $S_{l}$ be the set of all words of length $l$ not containing any word involving $1$ and $3$.
Let $S_{l2}$ be a subset of $S_{l}$ and $2$ belongs to each word in $S_{l2}$ .
Then $\{1\}\triangle W$, $W\in S_{l2}$, is of the form: $\{1,2\}\cup (W\setminus \{2\})$, where $(W\setminus \{2\})$ is of length $l-1$.
Similarly, $\{1\}\triangle W$, $W\in (S\setminus S_{l2})$, is of the form: $\{1\}\cup W$, where $W$ is of length $l$.
Similar argument can be made when the roles of $F_{1}$ and $F_{3}$, $F_{2}$ and $F_{4}$ are interchanged, respectively.
By the definition of $A_{l}^{(21)}$ and $A_{l}^{(22)}$, we obtain $K_{0l}(1)=A_{l-1}^{(2)}+A_{l}^{(2)}$.

For (c), first consider $\{2\}\triangle (\{1\}\cup W)$ and $\{2\}\triangle (\{1,2\}\cup W)$, where $W$ runs through all the words not involving $F_{1},F_{2},F_{3}$ and have length $l-1$.
It is equivalent to consider $\{1,2\}\cup W$ and $\{1\}\cup W$ for such $W$'s.
This yields $A_{l-1}^{(21)}$ in $K_{1l}(0)$.
Next we consider $\{j\}\triangle (\{1\}\cup W)$ and $\{j\}\triangle (\{1,2\}\cup W)$, where $j=4,...,n$ and $W$ runs through all the words not involving $F_{1},F_{2},F_{3}$ and have length $l-1$.
By \citet{Tang:99}, this yields $(n-l-1)A_{l-2}^{(21)}+lA_{l}^{(21)}$ in $K_{1l}(0)$.
Similar argument can be made when the roles of $F_{1}$ and $F_{3}$, $F_{2}$ and $F_{4}$ are interchanged, respectively.
By the definition of $A_{l}^{(21)}$ and $A_{l}^{(22)}$, we obtain $K_{1l}(0)=(n-l-1)A_{l-2}^{(2)} +A_{l-1}^{(2)}+lA_{l}^{(2)}$.

For (d), first consider $\{1\}\triangle (\{1\}\cup W)$, $\{1\}\triangle (\{1,2\}\cup W)$, $\{1,2\}\triangle (\{1\}\cup W)$ and $\{1,2\}\triangle (\{1,2\}\cup W)$, where $W$ runs through all the words not involving $F_{1},F_{2},F_{3}$ and have length $l-1$.
It is equivalent to consider $W$, $\{2\}\cup W$, $\{2\}\cup W$ and $W$ for such $W$'s.
This yields $2A_{l-1}^{(31)}$ in $K_{1l}(1)$.
Next consider $\{1\}\triangle (\{3\}\cup W)$, $\{1\}\triangle (\{3,4\}\cup W)$, $\{1,2\}\triangle (\{3\}\cup W)$ and $\{1,2\}\triangle (\{3,4\}\cup W)$, where $W$ runs through all the words not involving $F_{1},F_{3},F_{4}$ and have length $l-1$.
For such $W$'s, $\{1\}\triangle (\{3\}\cup W)$ and $\{1\}\triangle (\{3,4\}\cup W)$ yield $A_{l-1}^{(42)}$ in $K_{1l}(1)$;
$\{1,2\}\triangle (\{3\}\cup W)$ and $\{1,2\}\triangle (\{3,4\}\cup W)$ yield $A_{l-2}^{(43)}+A_{l-1}^{(52)}$ in $K_{1l}(1)$.
Similar argument can be made when the roles of $F_{1}$ and $F_{3}$, $F_{2}$ and $F_{4}$ are interchanged, respectively.
By the definition of $A_{l}^{(3)}$, we obtain $K_{1l}(1)=2A_{l-1}^{(3)}+2\{A_{l-1}^{(42)}+A_{l-2}^{(43)}+A_{l-1}^{(52)}\}$.

For (e), first consider $\{2\}\triangle (\{1,3\}\cup W)$, $\{2\}\triangle (\{1,3,4\}\cup W)$, $\{2\}\triangle (\{1,2,3\}\cup W)$, $\{2\}\triangle (\{1,2,3,4\}\cup W)$ and $\{4\}\triangle (\{1,3\}\cup W)$, $\{4\}\triangle (\{1,3,4\}\cup W)$, $\{4\}\triangle (\{1,2,3\}\cup W)$, $\{4\}\triangle (\{1,2,3,4\}\cup W)$, where $W$ runs through all the words not involving $F_{1},F_{2},F_{3},F_{4}$ and have length $l-2$.
This yields $2A_{l-2}^{(7)}$ in $K_{2l}(0)$.
Next consider $\{j\}\triangle (\{1,3\}\cup W)$, $\{j\}\triangle (\{1,3,4\}\cup W)$, $\{j\}\triangle (\{1,2,3\}\cup W)$, $\{j\}\triangle (\{1,2,3,4\}\cup W)$ for $j=5,...,n$.
By \citet{Tang:99}, this yields $(n-l-1)A_{l-3}^{(7)}+(l-1)A_{l-1}^{(7)}$.
So we obtain $K_{2l}(0)=2A_{l-2}^{(7)}+(n-l-1)A_{l-3}^{(7)}+(l-1)A_{l-1}^{(7)}$.

For (f), consider $\{j\}\triangle (\{1,3\}\cup W)$, $\{j\}\triangle (\{1,3,4\}\cup W)$, $\{j\}\triangle (\{1,2,3\}\cup W)$, $\{j\}\triangle (\{1,2,3,4\}\cup W)$ for $j=1,3$, and $\{i,j\}\triangle (\{1,3\}\cup W)$, $\{i,j\}\triangle (\{1,3,4\}\cup W)$, $\{i,j\}\triangle (\{1,2,3\}\cup W)$, $\{i,j\}\triangle (\{1,2,3,4\}\cup W)$ for $(i,j)=(1,2),(3,4)$, where $W$ runs through all the words not involving $F_{1},F_{2},F_{3},F_{4}$ and have length $l-2$.
This yields $2A_{l-2}^{(8)}$.
So we obtain $K_{2l}(1)=2A_{l-2}^{(8)}$.
\end{proof}
In view of Theorem 2, sequential minimization of $K$ is equivalent to that of the terms of $A=\{A_{3}^{(1)},A_{2}^{(2)},A_{1}^{(42)}+A_{1}^{(52)},A_{1}^{(7)},A_{4}^{(1)},A_{3}^{(2)},...\}$, which is reduced to $\{A_{3}^{(1)},A_{2}^{(2)},A_{1}^{(7)},A_{4}^{(1)},A_{3}^{(2)},...\}$ because $F_{1},F_{2},F_{3},F_{4}$ form a complete factorial, implying $A_{1}^{(42)}+A_{1}^{(52)}=A_{1}^{(7)}$.

We now develop a complementary set theory for the first four terms in $A$.
Let $\widetilde{T}$ be the complement of $\{b_{2},b_{4},...,b_{n}\}$ in $\Delta_{r}$; $A_{l}(\widetilde{T})$ be the number of ways of choosing $l$ members of $\widetilde{T}$ such that the sum of the chosen $l$ equals 0.
Let $T_{12}=\widetilde{T}\setminus \{b_1,b_{1}+b_{2}\}$; $T_{34}=\widetilde{T}\setminus \{b_3,b_{3}+b_{4}\}$; $A_{l}^{(12)}(T_{12})$ be the number of ways of choosing $l$ members of $T_{12}$ such that the sum of the chosen $l$ is in $\{b_{1},b_{1}+b_{2}\}$; $A_{l}^{(34)}(T_{34})$ be the number of ways of choosing $l$ members of $T_{34}$ such that the sum of the chosen $l$ is in $\{b_{3},b_{3}+b_{4}\}$.
Let $T=\Delta_{r}\setminus \{b_5,...,b_n\}$.
Define $H_{i}(\cdot,\cdot)$ as (2) in \citet{Mukerjee:01}.
\begin{thm}
Let $c_{j}$, $j=1,...,5$, be constants irrelevant to designs.
We have
\begin{enumerate}
\item [(a).] $A_{3}^{(1)}=c_{1}-A_{3}(\widetilde{T})$;
\item [(b).] $A_{4}^{(1)}=c_{2}+A_{3}(\widetilde{T})+A_{4}(\widetilde{T})$;
\item [(c).] $A_{2}^{(2)}=c_{3}+A_{2}^{(12)}(T_{12})+A_{2}^{(34)}(T_{34})$;
\item [(d).] $A_{1}^{(7)}=B_1+B_2+B_3+B_4$, where $B_1=c_{41}+H_{1}(\{b_1+b_3\},T)$ if $b_1+b_3=b_{j}$ for some $j\in\{5,...,n\}$ and zero otherwise; $B_2=c_{42}+H_{1}(\{b_1+b_2+b_3\},T)$ if $b_1+b_2+b_3=b_{j}$ for some $j\in\{5,...,n\}$ and zero otherwise; $B_3=c_{43}+H_{1}(\{b_1+b_3+b_4\},T)$ if $b_1+b_3+b_4=b_{j}$ for some $j\in\{5,...,n\}$ and zero otherwise; $B_4=c_{44}+H_{1}(\{b_1+b_2+b_3+b_4\},T)$ if $b_1+b_2+b_3+b_4=b_{j}$ for some $j\in\{5,...,n\}$ and zero otherwise. The $c_{4j}$'s are constants for every design.
\end{enumerate}
\end{thm}
\begin{proof}
Parts (a) and (b) are evident from \citet{Tang:96}.

For (c), note that $A_{2}^{(21)}=H_{2}(\{b_1\},\{b_4,...,b_n\})+H_{2}(\{b_1+b_2\},\{b_4,...,b_n\})$, which can be simplified as $A_{2}^{(21)}=c+H_{2}(\{b_1\},\{b_2,b_1+b_2\}\cup T_{12})+H_{2}(\{b_1+b_2\},\{b_1,b_2\}\cup T_{12})$ by Lemmas 1 and 3 in \citet{Mukerjee:01}, where $c$ is a constant for every design.
Because the design is an orthogonal array of strength two, we have $H_{2}(\{b_1\},\{b_2,b_1+b_2\}\cup T_{12})=1+H_{2}(\{b_1\},T_{12})$ and $H_{2}(\{b_1+b_2\},\{b_1,b_2\}\cup T_{12})=1+H_{2}(\{b_1+b_2\},T_{12})$.
Hence $A_{2}^{(21)}=c+2+H_{2}(\{b_1\},T_{12})+H_{2}(\{b_1+b_2\},T_{12})=c+2+A_{2}^{(12)}(T_{12})$.
Similarly, $A_{2}^{(22)}=c'+2+A_{2}^{(34)}(T_{34})$, where $c'$ is a constant for every design.
Therefore, we have $A_{2}^{(2)}=c_{3}+A_{2}^{(12)}(T_{12})+A_{2}^{(34)}(T_{34})$ by letting $c_3=c+c'+4$.

For (d), note that $A_{1}^{(7)}=H_{1}(\{b_1+b_3\},\{b_5,...,b_n\})+H_{1}(\{b_1+b_3+b_4\},\{b_5,...,b_n\})+H_{1}(\{b_1+b_2+b_3\},\{b_5,...,b_n\})+H_{1}(\{b_1+b_2+b_3+b_4\},\{b_5,...,b_n\})$.
Let $F=\Delta_{r}\setminus \{b_1+b_3,b_5,...,b_n\}$.
If $b_1+b_3\neq b_{j}$ for $j=5,...,n$, then $H_{1}(\{b_1+b_3\},\{b_5,...,b_n\})=0$.
If $b_1+b_3=b_{j}$ for some $j\in\{5,...,n\}$, then $H_{1}(\{b_1+b_3\},\{b_5,...,b_n\})=c_{41}+H_{1}(\{b_1+b_3\},F)$ by Lemmas 1 and 3 in \citet{Mukerjee:01}, where $c_{41}$ is a constant for every design.
Since $b_1+b_3=b_{j}$ for some $j\in\{5,...,n\}$, we have $F=T$ and $H_{1}(\{b_1+b_3\},F)=H_{1}(\{b_1+b_3\},T)$.
Thus $H_{1}(\{b_1+b_3\},\{b_5,...,b_n\})=B_1$.
Similarly, we have $H_{1}(\{b_1+b_2+b_3\},T)=B_2$, $H_{1}(\{b_1+b_3+b_4\},T)=B_3$ and $H_{1}(\{b_1+b_2+b_3+b_4\},T)=B_4$.
So, $A_{1}^{(7)}=B_1+B_2+B_3+B_4$.
\end{proof}

\section{An efficient computation procedure}\label{Sec: An efficient computation procedure}
We now develop a fast computational procedure which covers regular or nonregular designs.
For $0\leq c\leq n-2$, let $Q_{0}(c)=1$, $Q_{1}(c)=2c-(n-4)$, and
\begin{align}\label{Eq: recusive_Q}
Q_{l}(c)=l^{-1}\{[2c-(n-4)]Q_{l-1}(c)-(n-l-2)Q_{l-2}(c)\}
\end{align}
where $2\leq l\leq n-2$.
Write $\widetilde{\mathbf{D}}$ for the subarray given by the last $n-4$ columns of $\mathbf{D}$.
For $1\leq u,w\leq N$, let $c_{uw}$ be the number of positions where the $u$th and $w$th rows of $\widetilde{\mathbf{D}}$ have the same entry, and $q_{sl}(u,w)$ be the $(u,w)$th element of $\mathbf{X}_{sl}\mathbf{X}_{sl}^T$.
Denote the $(u,w)$th element of $\mathbf{D}$ by $d_{uw}$.
Then the following result holds.
\begin{thm}\label{Thm: fast_search}
For $1\leq u,w\leq N$ and $2\leq l\leq n-2$, we have
\begin{enumerate}
\item [(a).] $q_{0l}(u,w)=(d_{u2}d_{w2}d_{u4}d_{w4})Q_{l-2}(c_{u,w})+(d_{u2}d_{w2}+d_{u4}d_{w4})Q_{l-1}(c_{u,w})+Q_{l}(c_{u,w})$;
\item [(b).] $q_{1l}(u,w)=(d_{u1}d_{w1}+d_{u1}d_{w1}d_{u2}d_{w2}+d_{u3}d_{w3}+d_{u3}d_{w3}d_{u4}d_{w4})Q_{l-1}(c_{u,w})$;
\item [(c).] $q_{2l}(u,w)=d_{u1}d_{w1}d_{u3}d_{w3}(1+d_{u2}d_{w2}+d_{u4}d_{w4}+d_{u2}d_{w2}d_{u4}d_{w4})Q_{l-2}(c_{u,w})$.
\end{enumerate}
\end{thm}
\begin{proof}
For  $2\leq l\leq n-2$, let $\Sigma^{(l)}$ be the sum over binary tuples $j_{5}\cdots j_{n}$ such that $l$ of $j_{5},...,j_{n}$ equal $1$.
We have
\begin{align*}
q_{0l}(u,w)=&\Sigma^{(l)}\boldsymbol{x}(u;0000j_5\cdots j_n)\boldsymbol{x}(w;0000j_5\cdots j_n)\\
&+\Sigma^{(l-1)}\boldsymbol{x}(u;0100j_5\cdots j_n)\boldsymbol{x}(w;0100j_5\cdots j_n)\\
&+\Sigma^{(l-1)}\boldsymbol{x}(u;0001j_5\cdots j_n)\boldsymbol{x}(w;0001j_5\cdots j_n)\\
&+\Sigma^{(l-2)}\boldsymbol{x}(u;0101j_5\cdots j_n)\boldsymbol{x}(w;0101j_5\cdots j_n)\\
=&(d_{u2}d_{w2}d_{u4}d_{w4})\Psi_{l-2}(c_{u,w})+(d_{u2}d_{w2}+d_{u4}d_{w4})\Psi_{l-1}(c_{u,w})+\Psi_{l}(c_{u,w}),
\end{align*}
where $\Psi_{l}(u,w)=\Sigma^{(l)}\prod_{s=5}^{n}(d_{u}d_{w})^{j_s}$.
Similarly, we have
\begin{align*}
&q_{1l}(u,w)=(d_{u1}d_{w1}+d_{u1}d_{w1}d_{u2}d_{w2}+d_{u3}d_{w3}+d_{u3}d_{w3}d_{u4}d_{w4})\Psi_{l-1}(c_{u,w})\\
&q_{2l}(u,w)=d_{u1}d_{w1}d_{u3}d_{w3}(1+d_{u2}d_{w2}+d_{u4}d_{w4}+d_{u2}d_{w2}d_{u4}d_{w4})\Psi_{l-2}(c_{u,w}).
\end{align*}
The result will follow if $\Psi_{l}(u,w)=Q_{l}(c_{uw})$.
It is clear that $\Psi_{0}(u,w)=1$ and $\Psi_{1}(u,w)=c_{uw}+(-1)(n-4-c_{uw})=2c_{uw}-(n-4)$.
It remains to show $\Psi_{l}(u,w)$ satisfies the recursion relation (\ref{Eq: recusive_Q}).

Let $\Phi(\xi)=\prod_{j=5}^{n}(1+\xi d_{uj}d_{wj})$ and let $\Phi_{l}(\xi)$ be the $l$th derivative of $\Phi(\xi)$.
Note that $\Psi_{l}(u,w)=\Phi_{l}(0)/l!$.
Differentiation of $\log\Phi(\xi)$ yields
\begin{align*}
\Phi_{1}(\xi)&=\left( \sum_{j=5}^{n}\frac{d_{uj}d_{wj}}{1+\xi d_{uj}d_{wj}} \right) \Phi(\xi) \\
&=\left( \frac{c_{uw}}{1+\xi}-\frac{(n-4)-c_{uw}}{1-\xi} \right) \Phi(\xi),
\end{align*}
that is, $(1-\xi^2)\Phi_{1}(\xi)=\{ 2c_{uw}-(n-4)(1+\xi) \}\Phi(\xi)$.
Differentiating this $l-1$ and taking $\xi=0$, we get
\begin{align*}
\Phi_{l}(0)=[2c_{uw}-(n-4)]\Phi_{l-1}(0)-(l-1)(n-l-2)\Phi_{l-2}(0).
\end{align*}
This leads to (\ref{Eq: recusive_Q}) by using $\Psi_{l}(u,w)=\Phi_{l}(0)/l!$.
\end{proof}

With Theorem \ref{Thm: fast_search}, we can search for minimum aberration designs under conditional models through the following three steps:
\begin{enumerate}
\item [Step 1.] Start with a list of all nonisomorphic regular designs for given $N$ and $n$.
                         For $N=16$ and $32$, this can be done using the catalogs in \citet{Chen:93}.
\item [Step 2.] For each design in Step 1, permute its columns  such that the resulting design satisfies the conditions in Theorem \ref{Thm: universally_optimal}. Let the first four columns represent the two pairs of conditional and conditioned factors, that is, $F_{1},F_{2}$ and $F_{3},F_{4}$.
\item [Step 3.] For each design obtained in Step 2, calculate the sequence $K$ in (\ref{Seq_K}) by using Theorem \ref{Thm: fast_search}, and hence find a minimum aberration design.
\end{enumerate}

Table 1 exhibits the results of Steps 1-3 for $N=16$ and $5\leq n\leq 12$.
In the table, the numbers 1,2,4,8 represent basic factors in a design.
The other numbers represent added factors.
For example, for $n=5$, if the five factors are denoted by $A,B,C,D,E$, then the minimum aberration design is the one with the defining relation $E=ABCD$ because $15=1+2+4+8$.
We can see that all minimum aberration designs under conditional models are also minimum aberration under traditional models.
The finding supports using minimum aberration designs under traditional models to perform experiments in \citet{Su:15}.

Table 2 exhibits the results of Steps 1-3 for $N=32$ and $6\leq n\leq 18$.
In the table, the numbers 1,2,4,8,16 represent basic factors in a design.
The other numbers represent added factors.
For example, for $n=6$, if the five factors are denoted by $A,B,C,D,E,F$, then the minimum aberration design is the one with the defining relation $F=ABCDE$ because $31=1+2+4+8+15$.

\begin{table}[h]\label{Table: MA_N16}
\caption{Regular minimum aberration designs under conditional model for $N=16$} % title of Table
\centering
  \begin{tabular}{ | c | l |  }
    \hline
    $n$ & minimum aberration design \\ \hline
    5 & $(1,2,4,8,15)$  \\ \hline
    6 & $(1,8,2,4,7,11)$  \\ \hline
    7 &  $(1,2,4,8,7,11,13)$ \\ \hline
    8 &  $(1,2,4,8,7,11,13,14)$  \\ \hline
    9 &   $(2, 4, 8, 3, 1, 5, 9, 14, 15)$ \\ \hline
    10 &  $(1,6,2,8,4,3,5,9,14,15)$ \\ \hline
    11 &  $(4,8,5,10,1,2,3,6,9,13,14)$ \\ \hline
    12 &   $(2,5,6,10,1,4,8,3,9,13,14,15)$\\
    \hline
  \end{tabular}
\end{table}

\begin{table}[h]\label{Table: MA_N32}
\caption{Regular minimum aberration designs under conditional model for $N=32$} % title of Table
\centering
  \begin{tabular}{ | c | l |  }
    \hline
    $n$ & minimum aberration design \\ \hline
    6 & $(1,2,4,8,15,31)$  \\ \hline
    7 & $(1,8,16,7,2,4,27)$  \\ \hline
    8 &  $(4, 16, 7, 29, 1, 2, 8, 11)$ \\ \hline
    9 &  $(1, 4, 7, 29, 2, 8, 16, 11, 19)$  \\ \hline
    10 &   $(4,  8,  7,  19,  1,  2,  16,  11,  29, 30)$ \\ \hline
    11 &  $(16,  11,  14, 19,  1,  2,  4,  8,  7,  13, 21)^{*2}$ \\ \hline
    12 &  $(16,  11,  13, 19,  1,  2,  4,  8,  7,  14, 21, 22)^{*2}$ \\ \hline
    13 &  $(16,  11,  13, 19,  1,  2,  4,  8,  7,  14, 21, 22, 25)$ \\ \hline
    14 &  $(1,  4,  7,  11,  2,  8,  16,  13,  14, 19, 21, 22, 25, 26)$ \\ \hline
    15 &  $(1,  2,  4,  8,  16,  32,  7,  11,  13, 14, 19, 22, 25, 26, 28)$ \\ \hline
    16 &  $(1,  2,  4,  8,  16,  7,  11,  13,  14, 19, 21, 22, 25, 26, 28, 31)$\\  \hline
  \end{tabular}
\end{table}

\section{Concluding remarks}\label{Sec: Concluding remarks}
In this paper, we extend the work of \citet{Mukerjee:17} to two pairs of conditional and conditioned factors.
As mentioned in \citet{Mukerjee:17}, the number of conditional and conditioned pairs is seldom exceed two in practice.
Also, the effect hierarchy order is only irrelevant to the value of $r$ for the number of pairs not larger than two.
For more than two pairs, the effect hierarchy order is very complicated because different values of $r$ lead to different hierarchy.
We focus on regular designs in this paper.
Taking nonregular designs into consideration is left for future work.

\section*{Acknowledgments}
This research was supported by the Ministry of Science and Technology (Taiwan) via Grant Number 107-2118-M-008-001-MY2.

\bibliographystyle{plainnat}
\bibliography{D:/Dropbox/ref}
\end{document}